\ifx\documentclass\undefined
\documentstyle[12pt]{article}
\else
\documentclass[12pt]{article}
\fi

\usepackage{color}

\definecolor{redd}{rgb}{0.95,0.2,0.2}
\definecolor{gris}{rgb}{0.9,0.9,0.9}
\definecolor{greenn}{rgb}{0.1,0.6,0.2}

\definecolor{cmgray}{rgb}{0.7,0.7,0.7}
\newcommand{\tmgray}[1]{\textcolor{cmgray}{#1}}

\definecolor{cmblue}{rgb}{0.2,0.5,0.8}

\newcommand{\beq}{\begin{eqnarray*}}
\newcommand{\eeq}{\end{eqnarray*}}
\makeatletter
\renewcommand{\theequation}{\thesection.\arabic{equation}}
\@addtoreset{equation}{section}
\def\eqnarray{%
\stepcounter{equation}%
\let\@currentlabel=\theequation
\global\@eqnswtrue
\global\@eqcnt\z@
\tabskip\@centering
\let\\=\@eqncr
$$\halign to \displaywidth\bgroup\@eqnsel\hskip\@centering
$\displaystyle\tabskip\z@{##}$&\global\@eqcnt\@ne
\hfil$\displaystyle{{}##{}}$\hfil
&\global\@eqcnt\tw@$\displaystyle\tabskip\z@{##}$\hfil
\tabskip\@centering&\llap{##}\tabskip\z@\cr}
\makeatother
\newtheorem{theorem}{Theorem}[section]
\newtheorem{lemma}[theorem]{Lemma}

\newtheorem{proposition}[theorem]{Proposition}
\newtheorem{remark}{Remark}[section]
\newtheorem{definition}[theorem]{Definition}
\newsavebox{\toy}
\savebox{\toy}{\framebox[0.65em]{\rule{0cm}{1ex}}}
\newcommand{\QED}{\usebox{\toy}}
\def\nlni{\par\ifvmode\removelastskip\fi\vskip\baselineskip\noindent}
\newenvironment{proof}{\nlni\begingroup\it Proof.\rm}{
\endgroup\vskip\baselineskip}
\newcommand{\supp}{\mathop{\mathrm{supp}}\nolimits}

\newcommand{\esm}[1]{\mathbf{E}\left[#1\right]}

\begin{document}
%%%%%%% DOUBLE SPACED %%%%%%%%
\setlength{\baselineskip}{15pt}
\title{
Clock statistics for 1d Schr\"odinger operators
}
\author{
Victor Chulaevsky
\thanks{
D\'{e}partement de Math\'{e}matiques,
Universit\'{e} de Reims, Moulin de la Housse, B.P. 1039
51687 Reims Cedex 2, France.
E-mail: victor.tchoulaevski@univ-reims.fr
}
\and
Fumihiko Nakano
\thanks{
Department of Mathematics,
Gakushuin University,
1-5-1, Mejiro, Toshima-ku, Tokyo, 171-8588, Japan.
E-mail :
fumihiko@math.gakushuin.ac.jp}
}
%\date{最終更新日：}
\maketitle
%第一ページの番号を消す
%\thispagestyle{empty}
%%%%%%% ABSTRACT %%%%%%%%%%%%%
\begin{abstract}
We study
the 1d Schr\"odinger operators with alloy type random supercritical decaying potential and prove the clock convergence for the local statistics of eigenvalues.
We also consider,
besides the standard i.i.d. case, more general ones with exponentially decaying correlations.
\end{abstract}


%\tableofcontents
%%%%% INTRODUCTION %%%%%%%%%%%%%%%%%%%%%%%%%%%%%%%%%
\section{Introduction}
The level statistics
problem for 1d Schr\"odinger operators with random decaying potentials were studied by many
researchers, and various interesting results have been obtained
(cf., e.g., \cite{ALS}, \cite{KS}, \cite{KN1}, \cite{KN2}, \cite{MD}, \cite{N}).
Usually, one introduces local Hamiltonians 
$H_n$ 
in intervals of size 
$n$,
and  considers the point process 
$\xi_n$ 
generated by of the suitably rescaled eigenvalues of 
$H_n$.
Killip and Stoiciu \cite{KS}
showed that for CMV matrices weakly 
$\xi_n$ 
converges to the clock process,
limit of circular $\beta$-ensemble, and the Poisson process
for the supercritical, critical, and subcritical cases, respectively.
For the Schr\"odinger operator,
similar results are obtained by Avila \emph{et al.}
\cite{ALS} (supercritical discrete model), Krichevski \emph{et al.} \cite{KVV} (critical discrete model)
and by Kotani and Nakano \cite{KN1}, \cite{KN2}, \cite{N}
(the continuous model where the random potential is a function of the Browninan motion on a torus).
The aim of the present paper
is to prove the clock convergence (i.e., convergence to the clock process)
for the Schr\"odinger operator with the alloy type potential,

%
%%%To be explicit,
%%%we study
%%%the following operator.
%
\beq
H &=& - \frac {d^2}{dt^2} + V(t),
\quad
V(t) =
\sum_{j \in {\bf Z}}
\frac {\omega(j)}{j^{\alpha}}
f(t-j) \,,
\eeq
where
$\alpha > \frac 12$,
$f \in L^{\infty}$
with
$\supp f \subset [0, 1]$,
and the amplitudes 
$\omega_j$ 
of the potential on the cells $[j, j+1)$ are
i.i.d. or, more generally, form a stochastic process with exponentially decaying correlations.

Specifically, 
we consider one of the following two cases.
\\
{\bf A (i.i.d.): }
$\omega(j)$
are i.i.d., $|\omega(j)| \le 1$,
with ${\bf E}[ \omega (j) ] = 0$ and $\esm{\omega^2(j)}>0$.
Then by
\cite{KLS},
$H$
has a.s. purely a.c. spectrum on $[0, \infty)$.
\\
{\bf B (exponentially decaying correlation): }
$\{ \omega (j) \}_{j=1}^{\infty}$
is a bounded stochastic process such that
$| \omega (j) | \le 1$,
adapted to a filtration
$\{ {\cal F}_j \}_{j=1}^{\infty}$,
with exponentially decaying correlations :
\beq
\left|
{\bf E}[
\omega (j) | {\cal F}_k ]
\right|
\le
e^{- \rho |j - k|},
\quad
k < j,
\quad
\rho > 0.
\eeq
Clearly, {\bf A} is obtained as a special case of
{\bf B} by setting
${\cal F}_n$
to be the
$\sigma$-algebra
generated by
$\{ \omega (j) \}_{j=-\infty}^n$.
\begin{remark}
\label{ExampleB}
Examples of
\textbf{B} are provided,
e.g., by the following framework: let
$(\omega, {\cal F}, {\bf P}, T)$
be an ergodic dynamical system with discrete time
${\bf N}$
or
${\bf Z}$
admitting a finite Markov partition
$\Omega = C_1 \cup \cdots \cup C_M$,
$M>1$.
Pick a vector
$f = (f(1), \ldots, f(M) \in {\bf R}^M$
and
set
\beq
F(\theta)
&:=&
\sum_{i=1}^M f(i) 1_{C_i}(\theta) \,,
\\
\omega(j) &=& \omega(j,\theta) := F(T^j \theta) \,,
\eeq
where
$1_{C_i}(\cdot)$
is the indicator function of the partition element
$C_i$.
Then
the stochastic process
$\big\{ \omega(j) \big\}_{j=1}^\infty
=
\big\{ \omega(j,\theta) \big\}_{j=1}^\infty$
satisfies \textbf{B}.
A classical
example is given by
$\Omega = {\bf T}^1 = {\bf R}/{\bf Z}$
equipped with the Haar measure, and
$T$
is the dyadic transformation (an endomorphism)
$T:\, \theta \mapsto \{2\theta\}$
(here $\{\cdot\}$
stands for the fractional part).
Another
well-known example is the so-called natural extension of
the dyadic transformation (to an isomorphism), called Baker's transformation (see its definition in Section \ref{Baker}).
Furthermore,
there are infinitely many examples given by hyperbolic automorphisms of tori
${\bf T}^\nu$,
$\nu \ge 2$
which we also discuss in Section \ref{tori}.

In the present paper
we focus on the case of the most rapid (exponential) decay of correlations, and stress
the fact that the key arguments and estimates rely on bilinear, pair correlations, so they are applicable
even to stochastic processes emerging from deterministic dynamical systems, such as mentioned above.
There are
of course more "stochastic" (viz., non-deterministic) random processes to which our techniques
apply, for example, ergodic Markov chains on a finite or countable phase space, with discrete or continuous
time.
On the other hand,
recall that there are uncorrelated stochastic processes which are not i.i.d.
For example,
one can take the dyadic transformation and set
$\omega(j,\theta) = \cos(2^j \cdot 2\pi \theta)$
; then
${\bf E}[ \omega(j) \omega(k)]=0$
for
$j<k$
by orthogonality of the standard Fourier basis on
$[0, 2\pi]$,
yet
$\omega(j)$
and
$\omega(k)$
are not independent.
Indeed, taking
$j=1$, $k=2$,
we have
${\bf E}[\omega^2(1)]>0$,
${\bf E}[\omega(2)]=0$,
but
${\bf E}[\omega^2(1) \omega(2)]
=
(2\pi)^{-1} \int_0^{2\pi} \cos^2(\theta) \, \cos(2\theta)\, d\theta \ne 0$.

Finally,
note that the condition on the decay rate of pair correlation can be substantially relaxed.
\end{remark}
Let
$H_n$ be the Dirichlet restriction 
$H |_{[0, n]}$ 
of 
$H$ 
on 
$[0,n]$, 
with
$\{ E_j (n) \}_{j \ge j_0}$
being its positive eigenvalues, and
let
$\kappa_j(n) := \sqrt{E_j(n)}$.
Let
$E_0 > 0$
be the reference energy,
$\kappa_0 := \sqrt{E_0}$,
and consider the point process
\[
\xi_n :=
\sum_{j=1}^{\infty}
\delta_{n (\kappa_j(n) - \kappa_0)}.
\]
In the case 
of the free Laplacian, the atoms of 
$\xi_n$ 
are explicitly given by 
$\left\{ 
j \pi - n \kappa_0
\right\}_j$ 
so that to have the convergence of 
$\xi_n$, 
$n \kappa_0$ 
needs to converge up to 
$\pi$ 
: we have to consider a suitable subsequence 
$\xi_{n_k}$ 
of such point processes on intervals 
$[0, n_k]$. 
This is 
also the case in general which we henceforth assume except 
Theorem \ref{strongclock}.\\
\noindent
{\bf Assumption S} : 
A subsequence
$\{ n_k \}_{k=1}^{\infty}$
satisfies
\[
\kappa_0 n_k = m_k \pi + \beta + o(1),
\quad
k \to \infty
\]
for some
$m_k \in {\bf N}$
and
$\beta \in [0, \pi)$.
\begin{theorem}
\label{clock1}
%{\bf (Convergence to the clock process)}\\
%
Assume {\bf A} and {\bf S}.
Then
there exists a probability measure
$\mu_{\beta}$
on
$[0, \pi)$
such that
\[
\lim_{k \to \infty}
{\bf E}[ e^{- \xi_{n_k}} (g) ]
=
\int_0^{\pi}
d \mu_{\beta} (\phi)
\exp \left(
- \sum_{j \in {\bf Z}}
g( j \pi - \phi )
\right),
\quad
g \in C_c({\bf R}).
\]
\end{theorem}
\begin{theorem}
\label{clock2}
Assume
{\bf B} and {\bf S}.
Then
the statement of Theorem \ref{clock1} 
remains valid 
if we take subsequence of 
$\xi_{n_k}$ 
further. 
\end{theorem}
We believe that
the statement of Theorem \ref{clock2} is actually
true without taking subsequences.
For the moment, 
the problematic technical issue is
the lack of the Burkholder-Davies-Gundy
(BDG, in short) type inequality for the models with correlated amplitudes 
$\omega(j)$ (cf. Assumption {\bf B}).

Resorting to 
subsequences is not necessary, however, if we 
work with another formulation of the problem adopted in \cite{ALS}.
For given
$n$,
rearrange the eigenvalues
$\{ \kappa_k (n) \}$
of
$H_n$
in such a way that
\[
\cdots <
\kappa'_{-1}(n) <
\kappa'_{0}(n) < \kappa_0 < \kappa'_1(n) < \kappa'_2(n) <\cdots
\]
Then one has the following result. 
\begin{theorem}
{\bf (Strong clock behavior)}
\label{strongclock}
Assume {\bf A}.
We then have
\[
(\kappa'_{j+1}(n) - \kappa'_j(n)) n
\stackrel{n \to \infty}{\to} \pi,
\quad
j \in {\bf Z},
\quad
a.s.
\]
\end{theorem}
\begin{remark}
For the spectral property,
the argument
in this paper proves the following :
(1)
In case {\bf A},
$H$
has purely a.c. spectrum on
$[0, \infty)$
(as is shown in \cite{KLS})
(2)
In case {\bf B},
$\mu_{ac}(I) > 0$
for any interval
$I (\subset [0, \infty))$.
If
BDG inequality were true for case {\bf B}, we would have the same statement as in (1).
\end{remark}
\begin{remark}
We can also consider the
``decaying coupling constant model"
defined as follows.
\beq
H'_n :=
- \frac {d^2}{d t^2}
+
n^{- \alpha} V(t)
\quad
\mbox{on}
\;\;
L^2[0, n]
\eeq
with Dirichlet boundary condition.
Theorems \ref{clock1}, \ref{clock2}, and
\ref{strongclock}
are also valid for
$H'_n$,
except the fact that
$\phi = \beta$
is deterministic.
The proof
is simpler : for
$H'_n$
one can show
$n^{-2 \alpha}{\bf E}[ | J^{(n)} |^2 ]
\stackrel{n \to \infty}{\to} 0$
by using the method of proof of
Proposition \ref{square}.
\end{remark}
For the proofs of these theorems,
we basically follow the strategy of \cite{KS,KN1,KN2} :
to study the behavior of the relative Pr\"ufer phase
$\Theta$.
The clock convergence
essentially follows from the H\"older continuity of
$\Theta$
with respect to
$\kappa$,
after taking expectation.
Assuming {\bf A},
this is done by decomposing
$\Theta$
into the martingale part and the remainder.
Assuming {\bf B},
we use the ``conditioning argument" used in \cite{CS} to prove an extension of the martingale inequality and that of the maximal inequality, which is one of the main ingredient of this paper.

The rest
of this paper is organized as follows.
In Section 2,
we prepare some basic tools such as Pr\"ufer variables and obtain
a representation of the Laplace transform of the main point process
in terms of the relative Pr\"ufer phase
$\Theta$,
following the argument from \cite{KS}.
In Section 3,
we prove a version of martingale inequality assuming {\bf B}.
In Section 4,
we prove a version of the maximal inequality using the results in Section 3.
In Section 5,
we assume {\bf A} and prove the $p$-th power version of the results in Section 4,
by using the BDG inequality.
In Section 6,
we prove Theorems
\ref{clock1}, \ref{clock2}, \ref{strongclock}.
In Section 7,
a more detailed discussion
(continuation to Remark \ref{ExampleB})
is given on dynamical systems satisfying  {\bf B}.
Throughout
this paper,
$C$ stands for a positive constant which may change from line to line in each argument.
%

%%%%%%%%%%%%%%%%%%%%%%%%%%%%%%%%%%%%%%%
\section{Preliminaries}
Let
$H \psi = \kappa^2 \psi$,
$\psi(0) = 0$,
be a Schr\"odinger equation on 
$[0,+\infty)$ 
with the Dirichlet condition at $0$, which we rewrite
as a Cauchy problem for a vector-valued function,
\[
\left(
\begin{array}{c}
\psi \\ \psi'/\kappa
\end{array}
\right)
=
r_t (\kappa)
\left(
\begin{array}{c}
\sin \theta_t(\kappa) \\ \cos \theta_t(\kappa)
\end{array}
\right),
\quad
\theta_t(\kappa)
=:
\kappa t + \tilde{\theta}_t(\kappa).
\]
Then it follows by straightforward calculations that
\begin{eqnarray}
r_t(\kappa)
&=&
\exp \left(
\frac {1}{2\kappa} Im
\int_0^t
V(s) e^{2i \theta_s(\kappa)} ds
\right)
\label{r-eq}
\\
\tilde{\theta}_t (\kappa)
&=&
\frac {1}{2 \kappa}
\int_0^t
Re (e^{2i \theta_s(\kappa)} -1 )
V(s)
\label{integraleq}
\\
\frac {\partial \theta_t(\kappa)}{\partial \kappa}
&=&
\int_0^t
\frac {r_s^2}{r_t^2} ds
+
\frac {1}{2 \kappa^2}
\int_0^t
\frac {r_s^2}{r_t^2}
V(s)
(1 - Re \; e^{2i \theta_s(\kappa)}) ds.
\label{theta-kappa}
\end{eqnarray}
By 
Sturm's oscillation theorem, 
$j$-th 
eigenvalue 
$E_j(n)$ 
of 
$H_n$ 
satisfies 
$\theta_n \left(
\sqrt{E_j(n)}
\right) = j \pi$ 
by which we can derive the following representation of the Laplace transform of
$\xi_L$.
\begin{lemma}
\label{Laplace}
Let
\beq
\Theta^{(n)}(c)
&:=&
\theta_n(\kappa_c) - \theta_n(\kappa_0),
\\
\kappa_c &:=& \kappa_0 + \frac cn,
\quad
c \in {\bf R}
\eeq
then for
$g \in C_c({\bf R})$,
\beq
{\bf E}[ e^{- \xi_n(g)} ]
&=&
{\bf E}\left[
\exp \left(
- \sum_k
g \left(
(\Theta^{(n)})^{-1}
(
k \pi - \{\theta_n(\kappa_0) \}_{\pi {\bf Z}}
)
\right)
\right)
\right]
\eeq
where
\beq
\theta_n (\kappa)
&=&
[ \theta_n (\kappa) ]_{\pi}
\pi
+
\{ \theta_n (\kappa) \}_{\pi},
\quad
[x]_{\pi}
:=
\left\lfloor
\frac {x}{\pi}
\right\rfloor,
\quad
\{ x \}_{\pi}
:=
x - [x]_{\pi} \pi.
\eeq
\end{lemma}
By definition,
\beq
\Theta^{(n)}(c)
&=&
c +
\tilde{\theta}_n(\kappa_c)
-
\tilde{\theta}_n(\kappa_0)
\\
&=&
c
+
\frac {1}{2 \kappa_0}Re
\int_0^n
(e^{2i \theta_s(\kappa_c)}
-
e^{2i \theta_s(\kappa_0)})
V(s) ds
+
O(n^{-1}).
\eeq
In view of
Lemma \ref{Laplace},
the main task is to show that the 2nd term of RHS tends to
$0$.
To that end, we introduce the functional of the potential
\beq
J^{(t)} (\kappa)
:=
\int_0^t
e^{2i \theta_s(\kappa)} V(s) ds
\eeq
and prove the H\"older continuity of
$J^{(t)}(\kappa)$
with respect to
$\kappa$.
In order to do so,
we need the martingale and the maximal inequalities which we establish in the following sections.
%
%%%%%%%%%%%%%%%%%%%%%%%%%%%%%%%%

\section{Martingale inequality}
The strategy
of the proof of martingale inequality in case B
is based on a variant of the conditioning employed in
\cite{CS}
and the usual argument to prove the original martingale inequality.
\subsection{Notation and Statement}
In this section,
we work under a more general assumption B and set
\beq
J^{(m, N)} &:=&
\sum_{j=m}^{N-1} a_j \omega_j,
\quad
0 \le m \le N
\eeq
for some fixed 
$N$, 
where
$\{ \omega_k \}$
is the stochastic process satisfying the condition {\bf B} and
$a_j$
satisfies a measurability condition :
\begin{eqnarray}
&&
\left|
{\bf E}[ \omega_j | {\cal F}_k ]
\right|
\le
e^{- \rho |j-k|}
\label{decay}
\\
&&
a_j
:=
\frac {b_j}{j^{\alpha}},
\quad
b_j \in
{\cal F}_{j-c \log j},
\quad
c > 0.
\label{measurable}
%
%\| b \|_{\infty}&:=&\sup_{j \ge 1, \; \omega \in \Omega}| b_j (\omega) |
\end{eqnarray}
Here
we slightly abuse the notation and write
$c \log j$
instead of
$\left\lfloor c \log j \right\rfloor$.
The goal
of this section is to prove the following propositions.
\begin{proposition}
\label{square}
Suppose
\beq
| b_j (\omega) | \le c_j (\omega) \, j^{\eta},
\quad
\eta \ge 0,
\quad
\omega \in \Omega
\eeq
with
${\bf E}[ | c_j (\omega) |^2 ]
\le
C_{b, \eta}^2$.
Then
\beq
{\bf E}[ | J^{(m, N)} |^2]
& \le &
2C_{b, \eta}^2
\left(
\sum_{j=m}^{N-1}
\frac {c \log j}{
j^{\alpha - \eta} (j - c \log j)^{\alpha - \eta}
}
+
%C_{b, \eta}^2
\sum_{j=m}^{N-1}
\frac {1}
{j^{\alpha - 2\eta - 1 + c \rho}}
\right).
\eeq
\end{proposition}
\begin{proposition}
\label{martingale}
Suppose
\beq
| b_j (\omega) |
\le
c_j (\omega)
\, j^{\eta},
\quad
\eta \ge 0,
\quad
\omega \in \Omega.
\eeq
Then
\beq
{\bf E}\left[
\sup_{m \le n \le N}
\left|
J^{(m, n)}
\right|^p
\right]^{1/p}
& \le &
C_1
\left(
{\bf E}[ | J^{(m, N)} |^p ]^{1/p}
+
{\bf E}[ | D |^p ]^{1/p}
\right),
\quad
p > 1.
\\
\mbox{where }\;
D
&:=&
\left(
\sum_{j=m}^{N-1}
\frac {1}{j^{c \rho+\alpha-\eta}}
+
\frac {d \, m^{\beta}}{m^{\alpha - \eta}}
\right)
\sup_{m \le j \le N} c_j.
\eeq
\end{proposition}
%

%%%%%
\subsection{Proof of Proposition \ref{square}}
We decompose
the sum into two parts.
\beq
{\bf E}[ | J^{(m, N)} |^2]
&=&
2\sum_{m \le i \le j \le N}
{\bf E}[ a_i \omega_i a_j \omega_j ]
\\
&=&
2\sum_{j=m}^{N-1}
\sum_{k=0}^j
{\bf E} \left[
a_{j-k} \omega_{j-k} a_j \omega_j
\right]
\quad
(i = j-k)
\\
&=&
2\sum_{j=m}^{N-1}
\left(
\sum_{k=0}^{c \log j}
+
\sum_{k=c \log j}^j
\right)
{\bf E} \left[
a_{j-k} \omega_{j-k} a_j \omega_j
\right]
\\
&=:& I + II
\eeq
where
$I$
is the sum with
$|i -j| \le c \log j$,
and
$II$
is the remainder.

$I$
can be estimated easily:
\begin{eqnarray}
I
& \le &
2C_{b, \eta}^2
\sum_{j=m}^{N-1}
\sum_{k=0}^{c \log j}
\frac {1}{
j^{\alpha - \eta} (j - c \log j)^{\alpha - \eta}
}
\nonumber
\\
& \le &
2C_{b, \eta}^2
\sum_{j=m}^{N-1}
\frac {c \log j}{
j^{\alpha - \eta} (j - c \log j)^{\alpha - \eta}
}.
\label{squareone}
\end{eqnarray}
To estimate
$II$,
we use the condition
(\ref{decay}), (\ref{measurable})
: for
$k \ge c \log j$
we have
$\omega_{j-k} \in {\cal F}_{j-c \log j}$
so that
\beq
\left|
{\bf E}[ a_{j-k} a_j \omega_{j-k} \omega_j ]
\right|
&=&
\left|
{\bf E}[
a_{j-k} a_j \omega_{j-k}
{\bf E}[ \omega_j | {\cal F}_{j - c \log j} ]
]
\right|
\\
& \le &
{\bf E}
\Bigl[
\left|
a_{j-k} a_j \omega_{j-k}
\right|
\cdot
\left|
{\bf E}[ \omega_j | {\cal F}_{j - c \log j} ]
\right|
\Bigr]
\\
& \le &
e^{- \rho c \log j}
{\bf E}[ | a_{j-k} a_j \omega_{j-k} | ]
\eeq
which yields
\begin{eqnarray}
| II |
& \le &
2
\sum_{j=m}^{N-1}
\sum_{k = c \log j}^j
e^{- \rho c \log j}
{\bf E}
\Bigl[
| a_{j-k} a_j \omega_{j-k} |
\Bigr]
\nonumber
\\
& \le &
2C_{b, \eta}^2
\sum_{j=m}^{N-1}
\frac {1}
{j^{\alpha- 2\eta -1 + c \rho}}.
\label{squaretwo}
\end{eqnarray}

Using (\ref{squareone}) and (\ref{squaretwo})
completes the proof of Proposition \ref{square}.
\QED
%

%----- Martingale inequ. 2nd Form --------
%%%%%
\subsection{Proof of Proposition \ref{martingale}}
We decompose the sum such as
\beq
J^{(m, N)}
&:=&
\sum_{j=m}^{N-1} a_j \omega_j
=
J^{(m, n)} + J^{(n, N)},
\quad
m \le n \le N
\eeq
and the 2nd term in RHS is further decomposed into
\beq
J^{(n, N)}
&=&
\sum_{j=n}^{N-1} a_j \omega_j
=: J_{A_n} + J_{B_n}
\\
\mbox{ where }\;
A_n
&=&
\{ j \ge n \, | \, j - c \log j \ge n \},
\quad
B_n
=
\{ j \ge n \, | \, j - c \log j \le n \}.
\eeq
It is
easy to see that, for any
$\beta > 0$,
we can find
$d = d_{\beta} > 0$
such that
\beq
\sharp B_n \le d n^{\beta},
\quad
\forall \beta > 0.
\eeq
\noindent
(1)
For
$J_{A_n}$
we use (\ref{decay}), (\ref{measurable}) :
\beq
\left|
{\bf E}[ J_{A_n} | {\cal F}_n ]
\right|
& \le &
\sum_{j \in A_n}
{\bf E}\left[
|a_j| \cdot
\left|
{\bf E}[ \omega_j | {\cal F}_{j-c\log j} ]
\right|
\middle| {\cal F}_n
\right]
=
\sum_{j \in A_n}
\frac {1}{j^{c \rho +\alpha}}
{\bf E}[ | b_j | | {\cal F}_n ].
\eeq
(2)
For
$J_{B_n}$
we simply use the boundedness of
$\omega_j$ :
\beq
\left|
{\bf E}[ J_{B_n} | {\cal F}_n  ]
\right|
& \le &
\sum_{j \in B_n}
\frac {1}{j^{\alpha}}
{\bf E}[ | b_j  | | {\cal F}_n ].
\eeq
Therefore
\beq
\left|
{\bf E}[ J^{(n, N)} | {\cal F}_n ]
\right|
& \le &
\sum_{j \in A_n}
\frac {1}{j^{c \rho+\alpha}}
{\bf E}[ | b_j | | {\cal F}_n ]
+
\sum_{j \in B_n}
\frac {1}{j^{\alpha}}
{\bf E}[ | b_j  | | {\cal F}_n ].
\eeq
Since
$J^{(m, N)} = J^{(m, n)} + J^{(n, N)}$,
\begin{eqnarray}
&&
{\bf E}[ J^{(m, N)} | {\cal F}_n ]
=
J^{(m, n)} +
{\bf E}[ J^{(n, N)} | {\cal F}_n ]
\nonumber
\\
& \ge &
J^{(m, n)}
-
\left(
\sum_{j \in A_n}
\frac {1}{j^{c \rho + \alpha}}
{\bf E}[ | b_j | | {\cal F}_n ]
+
\sum_{j \in B_n}
\frac {1}{j^{\alpha}}
{\bf E}[ | b_j  | | {\cal F}_n ]
\right).
\label{lowerbound}
%\quad \cdots (*)
\end{eqnarray}
Let
\beq
T
&:=&
\left\{
\begin{array}{cc}
\inf
\left\{
n
\, \middle| \,
m \le n \le N, \;
J^{(m, n)} \ge \lambda
\right\}
&
\;
( J^{(m, n)} \ge \lambda
\mbox{ for some } n )
\\
N+1
&
\mbox{ (otherwise) }
\end{array}
\right.
\eeq
Then we have
\beq
\lambda
{\bf P}\left(
\sup_{n \le N}
J^{(m, n)} > \lambda
\right)
&=&
\lambda
\sum_{n=m}^N
{\bf P}(T = n)
\le
\sum_{n=m}^N
{\bf E} [ J^{(m, n)} \, ; \, T = n].
\eeq
Substituting
(\ref{lowerbound})
yields
\beq
&&
\lambda
{\bf P}\left(
\sup_{n \le N}
J^{(m,n)} > \lambda
\right)
\\
& \le &
\sum_{n=m}^N
{\bf E}
\left[
{\bf E}[ J^{(m,N)} | {\cal F}_n ]
+
\sum_{j \in A_n}
\frac {1}{j^{c \rho+ \alpha}}
{\bf E}[ | b_j | | {\cal F}_n ]
+
\sum_{j \in B_n}
\frac {1}{j^{\alpha}}
{\bf E}[ | b_j  | | {\cal F}_n ]
 \;;\;
T = n
\right]
\\
&=&
\sum_{n=m}^N
{\bf E}[ J^{(m,N)} \; ; \; T = n]
+
\sum_{n=m}^N
\sum_{j \in A_n}
\frac {1}{j^{c \rho+ \alpha}}
{\bf E}
\left[
| b_j | \, ; \, T = n
\right]
\\
&&  \qquad +
\sum_{n=m}^N
\sum_{j \in B_n}
\frac {1}{j^{\alpha}}
{\bf E}[ | b_j | ; T = n].
\eeq
Let
$
A :=
\{ \sup_{n \le N} J^{(m,n)} > \lambda \}.
$
We estimate
2nd and 3rd terms in RHS.
\\
(1)
2nd term :
\beq
&&
\sum_{n=m}^{N-1}
\sum_{j \in A_n}
\frac {1}{j^{c \rho}}
{\bf E}
\left[
| a_j | \, ; \, T = n
\right]
\\
&=&
\sum_{n=m}^{N-1}
\sum_{j=n}^{N-1}
1\left(
n \le j - c \log j
\right)
\frac {1}{j^{c \rho+ \alpha}}
{\bf E}\left[
| b_j | ; T = n, \; A
\right]
\\
&=&
\sum_{j=m}^{N-1}
\frac {1}{j^{c \rho+ \alpha}}
\sum_{n=m}^{j - c \log j}
%1 \left( n \le j - c \log j \right)
{\bf E}[ | b_j | ; T = n, \; A ]
\\
&\le&
\sum_{j=m}^{N-1}
\frac {1}{j^{c \rho+ \alpha}}
{\bf E} [ | b_j | \,;\, A ]
\\
& \le &
\left(
\sum_{j=m}^{N-1}
\frac {1}{j^{c \rho +\alpha-\eta}}
\right)
{\bf E} [ \sup_{m \le j \le N} c_j \,;\, A ].
\eeq
(2)
3rd term :
\beq
\sum_{n=m}^{N-1}
\sum_{j \in B_n}
\frac {1}{j^{\alpha}}
{\bf E}[ | b_j | \,;\, T = n]
& \le &
\sup_{n \ge m}
\left(
\frac {d n^{\beta}}{n^{\alpha-\eta}}
%\frac {d \, j^{\beta}}{j^{\alpha - \eta}}
\right)
\sum_{n=m}^{N-1}
{\bf E}[ \sup_{j \in B_n} c_j ; T=n, A]
\\
& \le &
\left(
\frac {d \, m^{\beta}}{m^{\alpha - \eta}}
\right)
{\bf E}[ \sup_{m \le j \le N} c_j ;  A],
\;\;
\beta < \alpha - \eta.
\eeq
Therefore
\beq
&&
\lambda
{\bf P}\left(
\sup_{n \le N}
J^{(m, n)} > \lambda
\right)
\le
{\bf E}\left[
J^{(m, N)}
+
D
 \,;\, A
\right].
\eeq
$D$
is defined in the statement of Proposition \ref{martingale}.
Let
$
\overline{J}^{(m, N)}
:=
\sup_{n \le N} J^{(m, n)}.
$
Then
\beq
{\bf E} \left[
| \overline{J}^{(m, N)} |^p
\right]
&=&
\int_0^{\infty}
p \lambda^{p-1}
{\bf P}(A) d \lambda
\\
& \le &
\int_0^{\infty} d \lambda
p \lambda^{p-2}
\int_{\Omega}
d {\bf P}
1_A (\omega)
\left(
J^{(m, N)} + D
\right)
\\
& \le &
\int_{\Omega}d {\bf P}
\left(
| J^{(m, N)}|+D
\right)
\int_0^{\infty}
d \lambda
p \lambda^{p-2}
1
\left(
\lambda < \overline{J}^{(m, N)}
\right)
\\
&=&
\frac {p}{p-1}
{\bf E}
\left[
\left(
| J^{(m, N)} | + D
\right)
( \overline{J}_N^{(m, N)} )^{p-1}
\right]
\\
& \le &
\frac {p}{p-1}
\left(
{\bf E}[ | J^{(m, N)} |^p ]^{1/p}
+
{\bf E}[ | D |^p ]^{1/p}
\right)
\cdot
{\bf E}[ (\overline{J}^{(m, N)})^p ]^{\frac {p-1}{p}}.
\eeq
Dividing
both sides by
${\bf E}[ (\overline{J}^{(m, N)})^p ]^{\frac {p-1}{p}}$
completes the proof.
\QED
%
%%%%%%%%%%%%%%%%%%%%%%%%%%%%%%%
\section{H\"older continuity}
In this section
we assume {\bf B}
and prove a version of the maximal inequalities for
$R$
and
$J$.
\subsection{Estimate on $R$}
Let
\beq
R^{(m, t)} (\kappa)
&:=&
\int_m^t
\left(
e^{2i \theta_s(\kappa)} - 1
\right)
V(s) ds \,,
\\
R^{(m,n)}(\kappa)
&=&
\sum_{j=m}^{n-1}
\frac {\omega(j)}{j^{\alpha}}
\int_j^{j+1}
\left(
e^{2i \kappa s + 2i \tilde{\theta}_s(\kappa)}-1
\right)
V(s) ds,
\quad
n \in {\bf N} \,,
\\
R^{(n)}(\kappa)
&:=&
R^{(0, n)} (\kappa).
\eeq
\begin{proposition}
\label{R2}
\beq
&&
{\bf E}\left[
\sup_{m \le t \le N} | R^{(m, t)}(\kappa) |^2
\right]
\\
& \le &
C
\Biggl\{
\sum_{j=m}^{N-1}
\frac {c \log j}{j^{\alpha}(j - c \log j)^{\alpha}}
+
\sum_{j=m}^{N-1}
\frac {1}{j^{\alpha-1+c\rho}}
+
\left(
\sum_{j=m}^{N-1}
\frac {1}{j^{c \rho+\alpha}}
\right)^2
+
\left(
\frac {d \, m^{\beta}}{m^{\alpha }}
\right)^2
\Biggr\}
\\
&& +
C
\left(
\sum_{j=m}^{N-1}
\frac {c \log j}{j^{2 \alpha}}
\right)^2.
\eeq
\end{proposition}
\begin{proof}
Fix $t\le N$ and set 
$n = \lfloor t \rfloor$.
Since
\beq
R^{(m,t)}
&=&
R^{(m, n)} + R^{(n, t)},
\quad
| R^{(n, t)} |
 \le
C n^{-\alpha},
\eeq
it suffices to estimate
${\bf E}[ \sup_{m \le n \le N}
| R^{(m, n)} |^2 ]$.
Here we set
\[
\hat{\theta}_j (\kappa)
:=
\tilde{\theta}_{j - c \log j} (\kappa),
\quad
c > 0.
\]
Then
we decompose
\begin{eqnarray}
&&
R^{(m, n)}(\kappa)
\nonumber
\\
&=&
\sum_{j=m}^{n-1}
\frac {\omega(j)}{j^{\alpha}}
\int_j^{j+1}
\Biggl\{
e^{2 i\kappa s + \hat{\theta}_j(\kappa)}
+
e^{2i \kappa s + 2i \hat{\theta}_j(\kappa)}
\left(
e^{2i
\left(
\tilde{\theta}_s(\kappa) - \hat{\theta}_j(\kappa)
\right)
}
-1
\right)
-1
\Biggr\}
f(s-j) ds
\nonumber
\\
&=:&
R_1^{(m,n)}(\kappa)
+
R_2^{(m,n)}(\kappa)
+
R_3^{(m,n)}(\kappa)
\label{decompositionR}
\end{eqnarray}
which we estimate separately. \\
(1)
$R_1$ :
Let
\beq
b_j :=
\int_j^{j+1}
e^{2i \kappa s} f(s-j) ds
\cdot
e^{2i \hat{\theta}_j(\kappa)}.
\eeq
By
Proposition \ref{square}
($\eta = 0$),
\beq
{\bf E}\left[
| R_1^{(m,n)} (\kappa) |^2
\right]
& \le &
C
\Biggl\{
\sum_{j=m}^{n-1}
\frac {c \log j}{j^{\alpha}(j - c \log j)^{\alpha}}
+
\sum_{j=m}^{n-1}
\frac {1}{j^{\alpha-1+c \rho}}
\Biggr\}.
\eeq
Thus by Proposition \ref{martingale},
\beq
&&
{\bf E}\left[
\sup_{m \le n \le N}| R_1^{(m, n)} (\kappa) |^2
\right]
\le
C_1
{\bf E}\left[
| R_1^{(m,N)} (\kappa) |^2
\right]
+
C_2
D^2
\quad
\\
& \le &
C
\Biggl\{
\sum_{j=m}^{N-1}
\frac {c \log j}{j^{\alpha}(j - c \log j)^{\alpha}}
+
\sum_{j=m}^{N-1}
\frac {1}{j^{\alpha-1+c \rho}}
+
\left(
\sum_{j=m}^{N-1}
\frac {1}{j^{c \rho+\alpha}}
\right)^2
+
\left(
\frac {d \, m^{\beta}}{m^{\alpha}}
\right)^2
\Biggr\}.
\eeq
(2)
$R_2$ : Let
\beq
b_j
&=&
\int_j^{j+1}
e^{2i \kappa s + 2i \hat{\theta}_j(\kappa)}
\left(
e^{2i
\left(
\tilde{\theta}_s(\kappa) - \hat{\theta}_j(\kappa)
\right)
}
-1
\right)
f(s-j) ds.
\eeq
Here we estimate
\beq
\left|
e^{2i
\left(
\tilde{\theta}_s(\kappa) - \hat{\theta}_j(\kappa)
\right)
}
-1
\right|
& \le &
2
\left|
\tilde{\theta}_s(\kappa) - \hat{\theta}_j(\kappa)
\right|
\\
& \le &
\frac {2}{2 \kappa}
\int_{j - c \log j}^{j+1}
\left|
\left(
e^{2i \theta_u(\kappa)}-1
\right)
V(u)
\right|
du
\\
& \le &
\frac {1 + c \log j}{(j - c \log j)^{\alpha}}
\le
\frac {c' \log j}{j^{\alpha}}
\eeq
for large
$j$
and for some
$c'$.
Hence we have
$| R_2^{(m, n)} (\kappa) |
\le
\sum_{j=m}^{n-1}
\frac {c' \log j}{j^{2 \alpha}}$
so that
\beq
{\bf E}
\left[
\sup_{m \le n \le N}
| R_2^{(m, n)} (\kappa) |^2
\right]
& \le &
C
\left(
\sum_{j=m}^{N-1}
\frac {c' \log j}{j^{2 \alpha}}
\right)^2.
\eeq
(3)
$R_3^{(m, n)}$ :
this is similar to that for
$R_1^{(m,n)}$.
\beq
&&
{\bf E}[ \sup_{n \le N} | R_3^{(m, n)} |^2 ]
\\
& \le &
C
\left(
\sum_{j=m}^{N-1}
\frac {c \log j}{
j^{\alpha } (j - c \log j)^{\alpha }
}
+
\sum_{j=m}^{N-1}
\frac {1}{j^{\alpha-1+c \rho}}
+
\left(
\sum_{j=m}^{N-1}
\frac {1}{j^{c \rho+\alpha}}
\right)^2
+
\left(
\frac {d \, m^{\beta}}
{m^{\alpha }}
\right)^2
\right).
\eeq
Putting
those estimates together, we complete the proof.
\QED
\end{proof}
%
%%%%%%%%%%%%%%%%%%%%%%%%%%%%%%%%%%%%%%
\subsection{Estimate on $J$}
%-------
For a function
$g = g(\kappa)$
of
$\kappa$,
we set
\[
\triangle g :=
g(\kappa_1) - g(\kappa_2),
\quad
\kappa_1, \kappa_2 \in {\bf R}.
\]
The goal
of this subsection is to prove the H\"older continuity of
$J$
in the following sense.
\begin{proposition}
\label{J2}
Let
$\eta > 0$
such that
$0 < \eta < \alpha - \frac 12$.
Then
\beq
&&
{\bf E}\left[
\sup_{m \le t \le N} | \triangle J^{(m, t)} |^2
\right]
\le
C
| \triangle \kappa |^{2 \eta}.
\eeq
\end{proposition}
%------------
The proof
is \tmgray{done} based on some ideas from \cite{KU}.
Set
$n = \lfloor t \rfloor$.
We decompose
\beq
J^{(m,t)}(\kappa)
&=:&
J^{(m, n)}(\kappa)
+
J^{(n, t)}(\kappa)
\\
\mbox{ where }\quad
J^{(m,n)}(\kappa)
&=&
\int_m^n
e^{2 i\theta_s(\kappa)} V(s) ds,
\quad
J^{(n, t)}(\kappa)
=
\int_n^t
e^{2 i\theta_s(\kappa)} V(s) ds.
\eeq
Out strategy is as follows:
we aim at proving the estimates
\begin{equation}
| \triangle J^{\sharp} |
\le
C_1
| \triangle J^{(m)} |
+
C_2
\sup_{m \le t \le N}
| \triangle J^{(m, t)} |
+
C_3
| \triangle \kappa |^{\eta} \,,
\label{strategy}
\end{equation}
where
$\sharp = (m, n)$ 
or 
$(n, t)$ 
and
$C_2 = o(1)$
as
$m \to \infty$.
%
%%%%%%%%%
\subsubsection{Estimate for
$\triangle J^{(n, t)}$}
\beq
\triangle J^{(n, t)}
&=&
\frac {\omega(n)}{n^{\alpha}}
\int_n^t
\Biggl\{
\left(
e^{2i \kappa_1 s} - e^{2i \kappa_2 s}
\right)
e^{2i\tilde{\theta}_s (\kappa_1)} f(s-n)
\\
&& \qquad
+
e^{2i \kappa_2 s}
\left(
e^{2i \tilde{\theta}_s(\kappa_1)}
-
e^{2i \tilde{\theta}_s(\kappa_2)}
\right)
f(s-n)
\Biggr\}
\\
&=:&
\triangle J^{(n, t)}_1 + \triangle J^{(n, t)}_2.
\eeq
Here we use
\begin{equation}
| e^{i \theta_1} - e^{i \theta_2} |
\le
C_{\eta}
| \theta_1 - \theta_2 |^{\eta},
\quad
\eta \in [0, 1] \,,
\label{theta}
\end{equation}
which yields
\beq
| \triangle J^{(n, t)}_1 |
& \le &
\frac {C}{n^{\alpha}}
n^{\eta}
| \triangle \kappa |^{\eta},
\quad
| \triangle J^{(n, t)}_2 |
 \le
\frac {C}{n^{\alpha}}
\int_n^t
| \triangle \tilde{\theta}_s | ds.
\eeq
For
$\triangle J_2$,
we estimate
$\tilde{\theta}_s$:
\beq
\triangle \tilde{\theta}_s
&=&
\frac {1}{2 \kappa_1}Re
\int_0^s
\left(
e^{2i \theta_u(\kappa_1)}-1
\right)
V(u)
du
-
\frac {1}{2 \kappa_2}Re
\int_0^s
\left(
e^{2i \theta_u(\kappa_2)}-1
\right)
V(u)
du
\\
&=&
\frac {1}{2 \kappa_1} Re
\int_0^s
\left(
e^{2i \theta_u(\kappa_1)}
-
e^{2i \theta_u(\kappa_2)}
\right)
V(u) du
\\
&& \quad+
\left(
\frac {1}{2 \kappa_1} - \frac {1}{2 \kappa_2}
\right)Re
\int_0^s
\left(
e^{2i \theta_u(\kappa_2)}-1
\right)V(u) du.
\eeq
Thus
\begin{equation}
| \triangle \tilde{\theta}_s |
\le
C
\left(
| \triangle J^{(m)} |
+
| \triangle J^{(m, s)} |
+
| \triangle \kappa |
| R^{(s)} (\kappa_2) |
\right)
\label{deltathetatilde}
\end{equation}
and therefore
\beq
| \triangle J_2^{(n,t)} |
& \le &
\frac {C}{n^{\alpha}}
\int_n^t
| \triangle \tilde{\theta}_s | ds
\\
& \le &
\frac {C}{n^{\alpha}}
\left(
| \triangle J^{(m)} |
+
\sup_{n \le s \le t}
| \triangle J^{(m,s)} |
+
| \triangle \kappa |
\sup_{n \le s \le t}
| R^{(s)} (\kappa_2) |
\right).
\eeq
Putting together,
we have
\begin{eqnarray}
&&
| \triangle J^{(n, t)} |
\le
\frac {1}{n^{\alpha - \eta}}
| \triangle \kappa |^{\eta}
\nonumber
\\
&&\qquad
+
\frac {C}{n^{\alpha}}
\left(
| \triangle J^{(m)} |
+
\sup_{n \le s \le t}
| \triangle J^{(m,s)} |
+
| \triangle \kappa |
\sup_{n \le s \le t}
| R^{(s)} (\kappa_2) |
\right).
\label{Jnt}
\end{eqnarray}
%
%%%%%%%%%%
\subsubsection{Estimate for $\triangle J^{(m,n)}$}
We next decompose
\beq
J^{(m,n)}(\kappa)
&=&
\sum_{j=m}^{n-1}
\frac {\omega(j)}{j^{\alpha}}
\int_j^{j+1}
e^{2i \kappa s + 2i \hat{\theta}_s(\kappa)} f(s-j) ds
\\
&&+
\sum_{j=m}^{n-1}
\frac {\omega(j)}{j^{\alpha}}
\int_j^{j+1}
e^{2i \kappa s + 2i \hat{\theta}_j(\kappa)}
\left(
e^{2i
(\tilde{\theta}_s(\kappa) - \hat{\theta}_j(\kappa))
}
-1
\right)
f(s-j) ds
\\
&=:&
J_1^{(m, n)}(\kappa) + J_2^{(m, n)}(\kappa) \,.
\eeq
The terms 
$J_1^{(m, n)}(\kappa)$ 
and 
$J_2^{(m, n)}(\kappa)$ 
will be estimated separately.
%
%%%%%
\paragraph{Estimate on
$J_{1}$.}

We
further decompose
\beq
\triangle J_1^{(m, n)}
&=&
\sum_{j=m}^{n-1}
\frac {\omega(j)}{j^{\alpha}}
\int_j^{j+1}
\left(
e^{2i \kappa_1 s} - e^{2i \kappa_2 s}
\right)
e^{2i \hat{\theta}_j (\kappa_1)}
f(s-j) ds
\\
&& +
\sum_{j=m}^{n-1}
\frac {\omega(j)}{j^{\alpha}}
\int_j^{j+1}
e^{2i \kappa_2 s}
\left(
e^{2i \hat{\theta}_j(\kappa_1)}
-
e^{2i \hat{\theta}_j(\kappa_2)}\right)
f(s-j) ds
\\
&=:&
\triangle J_{1-1}^{(m, n)}
+
\triangle J_{1-2}^{(m, n)}.
\eeq
(1)
$\triangle J_{1-1}$ : Let
\beq
b_j
&:=&
\int_j^{j+1}
\left(
e^{2i \kappa_1 s} - e^{2i \kappa_2 s}
\right)
e^{2i \hat{\theta}_j(\kappa_1)}
f(s-j) ds.
\eeq
Using
(\ref{theta}),
we have
$
| b_j | \le C_{\eta} | \triangle \kappa |^{\eta}
j^{\eta}.
$
By Proposition \ref{square} with
$c_j (\omega) = C| \triangle \kappa |^{\eta}$,
\beq
{\bf E}[ | J^{(m, N)} |^2]
& \le &
C
\left(
\sum_{j=m}^{N-1}
\frac {c \log j}{
j^{\alpha - \eta} (j - c \log j)^{\alpha - \eta}
}
+
%C_{b, \eta}^2
\sum_{j=m}^{N-1}
\frac {1}
{j^{\alpha-2 \eta -1 + c \rho}}
\right)
| \triangle \kappa |^{2 \eta}.
\eeq
Then by Proposition \ref{martingale} with
$c_j (\omega) = | \triangle \kappa |^{\eta}$,
\begin{eqnarray}
&&
{\bf E}\left[
\left|
\sup_{m \le n \le N}
\triangle J_{1-1}^{(m, n)}
\right|^2
\right]
 \le
C{\bf E}\left[
\left|
\triangle J_{1-1}^{(m, N)}
\right|^2
\right]
+
D^2
\nonumber
\\
&&
\le
C
\Biggl\{
\sum_{j=m}^{N-1}
\frac {c \log j}{
j^{\alpha - \eta} (j - c \log j)^{\alpha - \eta}
}
+
\sum_{j=m}^{N-1}
\frac {1}
{j^{\alpha-2 \eta -1 + c \rho}}
\Biggr\}
| \triangle \kappa |^{2\eta}
+
D^2
\label{J1-1}
\end{eqnarray}
where
\[
D
:=
\left(
\sum_{j=m}^{N-1}
\frac {1}{j^{c \rho+\alpha-\eta}}
+
\frac {d \, m^{\beta}}{m^{\alpha - \eta}}
\right)
| \triangle \kappa |^{\eta}.
\]
(2)
$\triangle J_{1-2}$ : Let
\beq
b_j
&=&
\int_j^{j+1}
e^{2i \kappa_2 s}
\left(
e^{2i \hat{\theta}_j(\kappa_1)}
-
e^{2i \hat{\theta}_j(\kappa_2)}
\right)
f(s-j) ds.
\eeq
By (\ref{deltathetatilde}),
\beq
| b_j |
\le
C | \triangle \hat{\theta}_j |
\le
C\left(
| \triangle J^{(m)} |
+
| \triangle J^{(m, j - c \log j)} |
+
| \triangle \kappa|
| R^{(j- c \log j)} (\kappa_2) |
\right).
\eeq
Thus by
Proposition \ref{square} with
$\eta = 0$,
we have
\begin{eqnarray}
&&
{\bf E}\left[
| \triangle J_{1-2}^{(m, n)} |^2
\right]
\le
\left(
\sum_{j=m}^{n-1}
\frac {1}{j^{\alpha-1+c \rho}}
+
\frac {c \log j}{ (j-c \log j)^{\alpha} j^{\alpha}}
\right)
\nonumber
\\
&&\qquad\quad
\times
{\bf E} \left[
| \triangle J^{(m)} |^2
+
\sup_{m \le j \le n} | J^{(m,j)} |^2
+
| \triangle \kappa |^2
\sup_{m \le j \le n} | R^{(j)} |^2
\right].
\label{J2one}
\end{eqnarray}
Using Proposition \ref{martingale} with
$\eta = 0$
and
\[
c_j (\omega)
=
| \triangle J^{(m)} |
+
| \triangle J^{(m, j - c \log j)} |
+
| \triangle \kappa|
| R^{(j- c \log j)} (\kappa_2) |,
\]
yields
\begin{equation}
{\bf E}\left[
\sup_{m \le n \le N} | \triangle J_{1-2}^{(m, n)} |^2
\right]
\le
C
{\bf E}\left[
| \triangle J_{1-2}^{(m, N)} |^2
\right]
+
{\bf E} [ |D|^2 ]
\label{J2two}
\end{equation}
where
\beq
D
&=&
\left(
\sum_{j=m}^{N-1}
\frac {1}{j^{c \rho+\alpha-\eta}}
+
\frac {d \, m^{\beta}}{m^{\alpha - \eta}}
\right)
\sup_{m \le j \le N} c_j (\omega).
\eeq
By
(\ref{J2one}), (\ref{J2two}),
we have
\begin{eqnarray}
&&
{\bf E}\left[
\sup_{m \le n \le N} | \triangle J_{1-2}^{(m, n)} |^2
\right]
\nonumber
\\
&&
\le
C
\Biggl\{
\sum_{j=m}^{N-1}
\left(
\frac {1}{j^{\alpha-1 + c \rho}}
+
\frac {c \log j}{ (j-c \log j)^{\alpha} j^{\alpha}}
\right)
+
\left(
\sum_{j=m}^{N-1}
\frac {1}{j^{c \rho+\alpha-\eta}}
\right)^2
+
\left(
\frac {d \, m^{\beta}}{m^{\alpha - \eta}}
\right)^2
\Biggr\}
\nonumber
\\
&& \qquad \times
{\bf E}\left[
| \triangle J^{(m)} |^2
+
\sup_{m \le j \le N}
| \triangle J^{(m, j)} |^2
+
| \triangle \kappa|^2
\sup_{m \le j \le N}
| R^{(j)} (\kappa_2) |^2
\right].
\label{J1-2}
\end{eqnarray}
%
%%%%%%%%%%%%%%%%
\paragraph{Estimate on
$J_{2}$.}
Let
$
\hat{E}_{s, j}
=
e^{2i
(\tilde{\theta}_s(\kappa) - \hat{\theta}_j(\kappa))
}
-1.
$
Then we have
\beq
\triangle J_2^{(m, n)}
&=&
\sum_{j=m}^{n-1}
\frac {\omega(j)}{j^{\alpha}}
\int_j^{j+1}
\Biggl\{
\left(
e^{2 i \kappa_1 s} - e^{2i \kappa_2 s}
\right)
e^{2i \hat{\theta}_j (\kappa_1) }
\hat{E}_{s, j}(\kappa_1)
\\
&&+
e^{2 \kappa_2 s}
\left(
e^{2i \hat{\theta}_j(\kappa_1)}
-
e^{2i \hat{\theta}_j(\kappa_2)}
\right)
\hat{E}_{s, j}(\kappa_1)
\\
&&+
e^{2i \kappa_2 s + 2i \hat{\theta}_j(\kappa_2)}
\left(
\hat{E}_{s, j}(\kappa_1)
-
\hat{E}_{s, j}(\kappa_2)
\right)
\Biggr\}
f(s-j) ds.
\eeq
Here we use
$
| \hat{E}_{s, j}(\kappa) |
\le
C j^{- \alpha} \log j
$
unifomrly w.r.t.
$\kappa$.
Moreover by
(\ref{integraleq}),
\beq
|
\triangle \hat{E}_{s, j}
|
&=&
\left|
e^{2i (\tilde{\theta}_s(\kappa_1) - \hat{\theta}_j(\kappa_1))}
-
e^{2i (\tilde{\theta}_s(\kappa_2) - \hat{\theta}_j(\kappa_2))}
\right|
\\
& \le &
C
\left|
(\tilde{\theta}_s(\kappa_1) - \hat{\theta}_j(\kappa_1))
-
(\tilde{\theta}_s(\kappa_2) - \hat{\theta}_j(\kappa_2))
\right|
\\
& \le &
C
\Biggl\{
\frac {c\log j}{(j - c \log j)^{\alpha}}
\cdot
j^{\eta} \cdot | \triangle \kappa |^{\eta}
+
\frac {c \log j}{(j - c \log j)^{\alpha}}
\sup_{j - c \log j \le u \le s}
| \triangle \tilde{\theta}_u |
\\
&&+
| \triangle \kappa |
\cdot
\frac {c \log j}{(j - c \log j)^{\alpha}}
\Biggr\}.
\eeq
Substituting the above bound,
we have
\begin{eqnarray}
&&
{\bf E}\left[
\sup_{m \le n \le N} | \triangle J_2^{(m,n)} |^2
\right]
\nonumber
\\
& \le &
\left\{
\sum_{j=m}^{n-1}
\left(
\frac {\log j}{j^{2 \alpha - \eta}}
+
\frac {c \log j}
{(j - c \log j)^{\alpha} j^{\alpha - \eta}}
+
\frac {c \log j}{(j - c \log j)^{\alpha} j^{\alpha}}
\right)
\right\}^2
| \triangle \kappa |^{2\eta}
\nonumber
\\
&&\quad+
\left\{
\sum_{j=m}^{n-1}
\left(
\frac {c \log j}{j^{2 \alpha}}
+
\frac {c \log j}{(j - c \log j)^{\alpha}j^{\alpha}}
\right)
\right\}^2
\times
\Biggl(
{\bf E}[ | \triangle J^{(m)}|^2]
\nonumber
\\
&&\quad+
{\bf E}[
| \sup_{m \le j \le N}| \triangle J^{(m, j)} |^2
]
+
| \triangle \kappa |^2
{\bf E}[
\sup_{0 \le j \le N}| R^{(j)}(\kappa_2) |^2
]
\Biggr).
\label{J22}
\end{eqnarray}
%
%%--------
\subsubsection{Proof of Proposition \ref{J2}}
By
(\ref{Jnt}), (\ref{J1-1}), (\ref{J1-2}) and
(\ref{J22}),
we have
\beq
&&
{\bf E}\left[
\sup_{m \le n \le N} | \triangle J^{(m, n)} |^2
\right]
\le
a_m
| \triangle \kappa |^{2 \eta}
\\
&&
+
b_m
\left(
{\bf E}[ | \triangle J^{(m)}|^2] +
{\bf E}[
| \sup_{m \le j \le N}| \triangle J^{(m, j)} |^2
]
+
| \triangle \kappa |^2
{\bf E}[
\sup_{j \le N}| R^{(j)}(\kappa_2) |^2
]
\right)
\eeq
where
$a_m,b_m = o(1)$,
as
$m \to \infty$.
Take
$m \gg 1$
s.t.
$b_m < 1$.
Then
\beq
&&
{\bf E}\left[
\sup_{m \le n \le N} | \triangle J^{(m, n)} |^2
\right]
\le
c_m
| \triangle \kappa |^{2 \eta}
+
d_m
\left(
{\bf E}[ | \triangle J^{(m)}|^2]
+
| \triangle \kappa |^2
{\bf E}[
\sup_{j \le N}| R^{(j)}(\kappa_2) |^2
]
\right)
\eeq
Here
we use the fact that
$
{\bf E}[ | \triangle J^{(m)}|^2]
\le
C
| \triangle \kappa |^2
$
(which follows from
(\ref{theta-kappa}))
and
Proposition \ref{R2},
completing the proof.
%

%%%%%%%%%%%%%%%%%%%%%%%%%%%%%%%%%
\section{Holder continuity : p-th power}
In this section,
we assume {\bf A} and prove estimates on the p-th power moment of
$R$
and
$J$
using BDG inequalities.
\subsection{Estimate on R}
\begin{proposition}
\label{Rp}
Assume {\bf A}.
Then
\beq
{\bf E}
\left[
\sup_{m \le t \le N}
| R^{(m, t)} (\kappa) |^p
\right]
\le
\left(
\sum_{j=m}^{N-1}
\frac {{\bf E}[\omega(j)^2]}{j^{2\alpha}}
\right)^{p/2}
+
\left(
\sum_{j=m}^{N-1}
\frac {1}{j^{\alpha}}
\right)^p.
\eeq
\end{proposition}
\begin{proof}
Set
$n = \lfloor t \rfloor$.
Then
\beq
R^{(m,t)}
&=&
R^{(m, n)} + R^{(n, t)}
\\
R^{(n, t)}
&=&
\int_n^t
\left(
e^{2i \theta_s(\kappa)} - 1
\right)
V(s) ds,
\quad
| R^{(n, t)} |
\le
C n^{- \alpha}
\eeq
so that it suffices to estimate
${\bf E}[ \sup_{m \le n \le N}
| R^{(m, n)} |^2 ]$.
We decompose
\beq
R^{(m,n)}(\kappa)
&:=&
\sum_{j=m}^{n-1}
\frac {\omega(j)}{j^{\alpha}}
\int_j^{j+1}
\Biggl\{
e^{2i \kappa s +
2i \tilde{\theta}_j(\kappa)}
\\
&& \quad
+
e^{2i \kappa s +
2i \tilde{\theta}_j(\kappa)}
\left(
e^{2i
(\tilde{\theta}_s(\kappa)
-
\tilde{\theta}_j(\kappa))
}
-1
\right)
-1
\Biggr\}
f(s-j) ds
\\
&=:&
R_1^{(m,n)}(\kappa)
+
R_2^{(m,n)}(\kappa)
+
R_3^{(m,n)}(\kappa).
\eeq
which is slightly different from
(\ref{decompositionR}). \\
(1)
$R_1^{(m,n)} (\kappa)$ :
\beq
{\bf E}[ | R_1^{(n)}(\kappa) |^2 ]
&=&
\sum_{j=m}^{n-1}
{\bf E}
\Biggl[
\frac {\omega(j)^2}{j^{2\alpha}}
\Biggl|
\int_j^{j+1}
e^{2i \kappa s + 2i
\tilde{\theta}_j(\kappa)}
f(s-j) ds
\Biggr|^2
\Biggr]
%\\
%
%&\le&
\le
C
\sum_{j=m}^{n-1}
\frac {{\bf E}[\omega(j)^2]}{j^{2\alpha}}.
\eeq
By BDG,
\begin{equation}
{\bf E}\left[
\sup_{m \le j \le N}
| R_1^{(m, j)}(\kappa)|^p
\right]
\le
C
{\bf E} \left[
| R_1^{(m, N)}(\kappa) |^2
\right]^{p/2}
=
C
\left(
\sum_{j=m}^{N-1}
\frac {{\bf E}[\omega(j)^2]}{j^{2\alpha}}
\right)^{p/2}.
\label{Rpone}
\end{equation}
(2)
$R_2^{(m,n)}(\kappa)$ :
Here we use
\beq
\left|
e^{2i (
\tilde{\theta}_s(\kappa)
-
\tilde{\theta}_j(\kappa)
)}
-1
\right|
& \le &
2
| \tilde{\theta}_s(\kappa)
-
\tilde{\theta}_j(\kappa) |
\le
C j^{- \alpha}
\eeq
and the fact that
$\omega(j)$
is bounded, yielding
\begin{equation}
{\bf E} \left[
\sup_{m \le n \le N}
| R_2^{(m,n)}(\kappa) |^p
\right]
\le
C
\left(
\sum_{j=m}^{N-1}
\frac {1}{j^{2\alpha}}
\right)^p.
\label{Rptwo}
\end{equation}
(3)
$R_3^{(m,n)}(\kappa)$ :
this is similar to (1) above :
\begin{equation}
{\bf E}\left[
\sup_{m \le j \le N}
| R_3^{(m, j)} |^p
\right]
\le
{\bf E}\left[
| R_3^{(m, N)} |^2
\right]^{p/2}
\le
\left(
\sum_{j=m}^{N-1}
\frac {
{\bf E}[\omega(j)^2]
}
{ j^{2 \alpha} }
\right)^{p/2}.
\label{Rpthree}
\end{equation}
By (\ref{Rpone}),
(\ref{Rptwo}), and
(\ref{Rpthree}),
we complete the proof.
\QED
\end{proof}
%

%%%%%%%%%%%%%%%%%%%%%%%%%%%%%%%%%
\subsection{Estimate on J}
\begin{proposition}
\label{Jp}
Assume {\bf A}.
Then
\beq
{\bf E}\left[
\sup_{m \le t \le N}
| J^{(m, t)} |^p
\right]
\le
| \triangle \kappa |^{p \eta},
\quad
p \ge 2.
\eeq
\end{proposition}
\begin{proof}
Set
$n = \lfloor t \rfloor$.
Then
we decompose
\beq
J^{(m, t)}(\kappa)
&=:&
J^{(m, n)}(\kappa)
+
J_0^{(n, t)}(\kappa)
\eeq
and
\beq
J^{(m, n)}(\kappa)
&=&
\sum_{j=m}^{n-1}
\frac {\omega(j)}{j^{\alpha}}
\int_j^{j+1}
e^{2i \kappa s +
2i \tilde{\theta}_j(\kappa)}
f(s-j) ds
\\
&& \quad +
\sum_{j=m}^{n-1}
\frac {\omega(j)}{j^{\alpha}}
\int_j^{j+1}
e^{2i \kappa s +
2i \tilde{\theta}_j(\kappa)}
\left(
e^{2i
(\tilde{\theta}_s(\kappa)
-
\tilde{\theta}_j(\kappa))
}
-1
\right)
f(s-j)
\\
&=:&
J_1^{(m, n)}(\kappa)
+
J_2^{(m, n)}(\kappa).
\eeq
Out strategy is
to obtain inequalities similar to
(\ref{strategy}).
\\

\noindent
(0) Estimate for
$\triangle J_0$ :
this can be done as
(\ref{Jnt}).
\beq
| \triangle J_0 |
& \le &
\frac {1}{n^{\alpha - \eta}}
| \triangle \kappa |^{\eta}
+
\frac {1}{n^{\alpha}}
\Biggl(
\sup_{s \le t} | \triangle J^{(s)} |
+
| \triangle \kappa |
\sup_{s \le t} | R^{(s)} |
\Biggr).
\eeq
(1)Estimate for
$\triangle J_1$ :
By the argument used
in the proof of Proposition \ref{J2},
we have
\beq
&&
{\bf E}\left[
| \triangle J_{1}^{(m, N)} |^2
\right]
\le
\sum_{j=m}^{N-1}
\frac {{\bf E}[\omega(j)^2]}{j^{2\alpha-2\eta}}
\\
&& \times
\Biggl(
| \triangle \kappa |^{2 \eta}
+
{\bf E}
[| \triangle J^{(m)} |^2]
+
{\bf E}[
\sup_{m \le j \le N}
| \triangle J^{(m, j)} |^2]
+
| \triangle \kappa |^2
{\bf E}[
\sup_{m \le j \le N}
| R^{(j)} (\kappa_2) |^2]
\Biggr).
\eeq
Let
$p \ge 2$.
Taking
$p/2$-th power on both sides,
\beq
&&
{\bf E}\left[
| \triangle J_{1}^{(m, N)} |^2
\right]^{p/2}
\le
\left(
\sum_{j=m}^{N-1}
\frac {{\bf E}[\omega(j)^2]}{j^{2\alpha-2\eta}}
\right)^{p/2}
%
%\times
\Biggl\{
| \triangle \kappa |^{p\eta}
+
{\bf E}
[ | \triangle J^{(m)} |^2 ]^{p/2}
\\
&&\qquad\qquad\quad
+
{\bf E}\left[
\sup_{m \le j \le N}
| \triangle J^{(m, j)} |^2
\right]^{p/2}
%\\
%
%&& \qquad
+
| \triangle \kappa |^{p}
{\bf E}[
\sup_{m \le j \le N}
| R^{(j)} (\kappa_2) |^2]^{p/2}
\Biggr\}.
\eeq
Here we use the following facts. \\

\noindent
(i)
Burkholder-Davies-Gundy inequality (BDG) :
\beq
{\bf E}\left[
| \triangle J_1^{(m, N)} |^2
\right]^{p/2}
\simeq
{\bf E}
\left[
\sup_{m \le j \le N}
| \triangle J_1^{(m, j)} |^p
\right]
\eeq
(ii)
Using
$
{\bf E}[ |X| ]
\le
{\bf E}[ |X|^p ]^{1/p}
$
valid for
$p \ge 1$,
we have
\beq
{\bf E}\left[
\sup_{m \le j \le N}
\left|
\triangle J^{(m, j)}
\right|^2
\right]
&=&
{\bf E}\left[
\left|
\sup_{m \le j \le N}
\triangle J^{(m, j)}
\right|^2
\right]
\le
{\bf E}\left[
\left|
\sup_{m \le j \le N}
\triangle J^{(m, j)}
\right|^p
\right]^{2/p}.
\eeq
Using (i), (ii) above,
we have
\begin{eqnarray}
&&
{\bf E}\left[
\sup_{m \le j \le N}
| \triangle J_1^{(m, j)} |^p
\right]
\le
\left(
\sum_{j=m}^{N-1}
\frac {{\bf E}[\omega(j)^2]}{j^{2\alpha-2\eta}}
\right)^{p/2}
%
%\times
\Biggl\{
| \triangle \kappa |^{p\eta}
+
{\bf E}
[ | \triangle J^{(m)} |^p ]
\nonumber
\\
&& \qquad\qquad
+
{\bf E}\left[
\left|
\sup_{m \le j\le N}
\triangle J^{(m, j)}
\right|^p
\right]
+
| \triangle \kappa |^{p}
{\bf E}[
\sup_{m \le j \le N}
| R^{(j)} (\kappa_2) |^2]^{p/2}
\Biggr\}.
\label{Jpone}
\end{eqnarray}
(2)
Estimate for
$\triangle J_2$ :
by the argument
in the proof of Proposition
\ref{J2},
\beq
&&
\sup_{m \le j \le N}
| \triangle J_2^{(m, j)} |^2
\le
\left(
\sum_{j=m}^{N-1}
\frac {1}{j^{2 \alpha - \eta}}
\right)^2
\Biggl\{
| \triangle \kappa |^{2 \eta}
+
| \triangle J^{(m)} |^2
\\
&& \qquad\qquad\qquad
+
\sup_{m \le j \le N}
| \triangle J^{(m, j)} |^2
+
| \triangle \kappa |^2
\sup_{m \le j \le N}
| R^{(j)} (\kappa_2) |^2
\Biggr\}.
\eeq
Take the
$p/2$-th power,
and then expectation.
\begin{eqnarray}
&&
{\bf E}\left[
\sup_{m \le j \le N}
| \triangle J_2^{(m, j)} |^p
\right]
\le
\left(
\sum_{j=m}^{N-1}
\frac {1}{j^{2 \alpha - \eta}}
\right)^p
\Biggl\{
| \triangle \kappa |^{p\eta}
+
{\bf E}\left[
| \triangle J^{(m)} |^p
\right]
\nonumber
\\
&&\qquad
+
{\bf E}\left[
\sup_{m \le j \le N}
| \triangle J^{(m, j)} |^p
\right]
+
| \triangle \kappa |^p
{\bf E}\left[
\sup_{m \le j \le N}
| R^{(j)} (\kappa_2) |^p
\right]
\Biggr\}.
\label{Jptwo}
\end{eqnarray}
(3)
Putting together :
now the rest of the argument is quite similar to that of Proposition \ref{J2},
so that we omit the details.
\QED
\end{proof}
%

%
%%%%%%%%%%%%%%%%%%%%%%%%%%%%%%%
\section{Proof of Theorems}
The following
two propositions are the key ingredients of the proof of the clock convergence.
\begin{proposition}
\label{Theta}
(1)
Assume {\bf A}.
We then have
\[
\Theta^{(n)}(c)
\to c,
\quad
\mbox{for all}
\quad
c \in {\bf R},
\quad
a.s.
\]
(2)
Assume {\bf B}
and
let
$\beta > 0$
satisfies
$2 \eta \beta > 1$.
Then
\[
\Theta^{(k^{\beta})}(c)
\to c,
\quad
\mbox{for all}
\quad
c \in {\bf R},
\quad
a.s.
\]
\end{proposition}
\begin{proof}
\mbox{}\\
Proof of (2) :
Assume {\bf B}. 
Then 
\beq
\Theta^{(N)}(c) -c
&=&
\frac {1}{2 \kappa_c}
Re \;
\triangle J^{(N)}
+
\left(
\frac {1}{2 \kappa_c}
-
\frac {1}{2 \kappa_0}
\right)
Re \;
R^{(N)}(\kappa_0).
\eeq
By Propositions
\ref{R2}, \ref{J2},
we have
\beq
{\bf P}\left(
| \Theta^{(N)}(c) - c | \ge \epsilon
\right)
& \le &
\frac {C}{\epsilon^2}
{\bf E}
\left[
| \triangle J^{(N)} |^2
+
| \triangle \kappa |^2
| R^{(N)}(\kappa_0) |^2
\right]
\\
& \le &
\frac {C}{\epsilon^2}
\left(
| \triangle \kappa |^{2 \eta}
+
| \triangle \kappa |^2
\right)
\eeq
where we put
$\triangle\kappa
=
\kappa_c - \kappa_0
=
\frac cn$.
Take
$\beta > 0$
such that
$
2 \eta \beta > 1
$
and consider a subsequence of
$N$ ;
$N := k^{\beta}$.
Let
\[
B_{k, \epsilon}(c)
:=
\left\{
| \Theta^{(k^{\beta})}(c) - c |
\ge
\epsilon
\right\},
\quad
k = 1, 2, \cdots,
\quad
\epsilon > 0.
\]
Then 
by the Borel--Cantelli lemma,
${\bf P}\left(
\limsup_{k \to \infty} B_{k, \epsilon}(c)
\right) = 0$,
so that
\[
{\bf P}\left(
\bigcap_{l \ge 1}
\limsup_{k \to \infty} B_{k, \frac 1l}(c)
\right) = 0.
\]
Therefore
$\Theta^{(k^{\beta})}(c) \to c$,
a.s.
for any fixed
$c \in {\bf R}$.
Since
$h(c) := c$
is continuous ant non-decreasing,
\[
\Theta^{(k^{\beta})}(c)
\stackrel{k \to \infty}{\to} c
\]
on the compliment of the event
$\left(
\bigcup_{c \in {\bf Q}}
\bigcap_{l \ge 1}
\limsup_{k \to \infty} B_{k, \frac 1l}(c)
\right)$.\\
\noindent
Proof of (1)
Assume {\bf A}.
Then
Proposition \ref{Jp} yields
\beq
{\bf P}\left(
| \Theta^{(N)}(c) - c | \ge \epsilon
\right)
& \le &
\frac {C}{\epsilon^p}
{\bf E}
\left[
| \triangle J^{(N)} |^p
+
| \triangle \kappa |^p
| R^{(N)}(\kappa_0) |^p
\right]
\\
& \le &
\frac {C}{\epsilon^p}
\left(
| \triangle \kappa |^{p \eta}
+
| \triangle \kappa |^p
\right).
\eeq
Taking
$p \gg 1$
s.t.
$p \eta > 1$
would give us the a.s. convergence
without taking further subsequence.
\QED
\end{proof}
\begin{proposition}
\label{thetatilde}
For any fixed
$\kappa$,
\[
\tilde{\theta}_n (\kappa)
\stackrel{n \to \infty}{\to}
\tilde{\theta}_{\infty}(\kappa).
\]
\end{proposition}
\begin{proof}
Since
$\tilde{\theta}_n (\kappa)
=
\frac {1}{2 \kappa}
Re \; R^{(n)}(\kappa)$,
it suffices to show the convergence of
$R^{(n)}(\kappa)$.
By Proposition \ref{R2},
\beq
{\bf E} \left[
\sup_{2^k \le n \le 2^{k+1}}
| R^{(2^k, n)}(\kappa) |^2
\right]
&\le&
C
\sum_{2^k \le j \le 2^{k+1}}
\frac {1}{j^{2 \alpha - \epsilon}}
+
\left(
\sum_{2^k \le j \le 2^{k+1}}
\frac {1}{j^{2 \alpha - \epsilon}}
\right)^2
+
\frac {1}{
2^{2(\alpha - \eta - \beta)k}}
\\
& \le &
\frac {C}{2^{(2 \alpha-1-\epsilon)k}}
\eeq
for
$\epsilon > 0$
sufficiently small, by taking
$\eta+ \beta$
sufficiently small.
By
Chebyshev's inequality,
\beq
{\bf P}\left(
\sup_{2^k \le n \le 2^{k+1}}
| R^{(2^k, n)}(\kappa) |^2
\ge
\frac {1}{k^4}
\right)
& \le &
C
\frac {k^4}{2^{(2 \alpha-1 - \epsilon)k}}.
\eeq
By the Borel-Cantelli lemma,
with probability one
\beq
\sup_{2^k \le n \le 2^{k+1}}
| R^{(2^k, n)}(\kappa) |
\le
\frac {1}{k^2}
\eeq
for
$k \gg 1$,
implying that
$\{ R^{(n)}(\kappa)\}_n$
is Cauchy.
\QED
\end{proof}
\begin{lemma}
\label{inverse}
{\bf (\cite{KN1}, Lemma 3.3)}
Let
$\Psi_n, n=1,2, \cdots$,
and
$\Psi$
are continuous and increasing functions such that
$\lim_{n \to \infty}\Psi_n(x)=\Psi(x)$
pointwise.
If
$y_n \in Ran \Psi_n$,
$y \in Ran \Psi$
and
$y_n \to y$,
then it holds that
\[
\Psi_n^{-1}(y_n)
\stackrel{n \to \infty}{\to}
\Psi^{-1}(y).
\]
\end{lemma}
{\it Proof of Theorems \ref{clock1}, \ref{clock2}}\\
In the representation
of the Laplace transform of
$\xi_n$
(Lemma \ref{Laplace}),
we use
Propositions \ref{Theta}, \ref{thetatilde},
and
Lemma \ref{inverse}.
\QED\\
%

%%%%%%
For the proof of Theorem \ref{strongclock},
for given
$n$,
rearrange the eigenvalues
$\{ \kappa_k (n) \}$
of
$H_n$
such that
\[
\cdots <
\kappa'_{0}(n) < \kappa_0 < \kappa'_1(n) < \cdots
\]
\begin{lemma}
\label{estimate}
For any fixed
$k$,
$\kappa'_k(n)
=
\kappa_0 + o(1),
\quad
n \to \infty$.
\end{lemma}
\begin{proof}
By definition of
$\kappa'_k(n)$,
\begin{equation}
(m_n+ k) \pi
=
\theta_n (\kappa'_k(n))
=
\kappa'_k(n) n + \tilde{\theta}_n (\kappa'_k(n))
\label{star}
\end{equation}
Write
$\kappa_0 n = m_n \pi + \beta_n$,
$m_n \in {\bf N}$,
$\beta_n \in [0, \pi)$.
Substituting
$m_n \pi = \kappa_0 n - \beta_n$
into
$(\ref{star})$
yields
\beq
\kappa'_k(n)
=
\kappa_0
+
\frac {
- \beta_n + k \pi - \tilde{\theta_n}(\kappa'_k(n))
}
{n}
\eeq
By
(\ref{integraleq}),
$n^{-1} \tilde{\theta}_n (\kappa'_k(n)) \to 0$.
\QED
\end{proof}
{\it Proof of Theorem \ref{strongclock}}
\beq
\theta_n (\kappa'_j(n))
&=&
[ \theta_n (\kappa_0) ]_{\pi}
+ j \pi
\\
\theta_n (\kappa'_{j+1}(n))
&=&
[ \theta_n (\kappa_0) ]_{\pi}
+ (j+1) \pi
\eeq
from which we have
\beq
(\kappa'_{j+1}(n) - \kappa'_j(n))n
-
(\tilde{\theta}_{n}(\kappa'_{j+1}(n)) -\tilde{\theta}_{n}(\kappa'_j(n)))
=
\pi.
\eeq
%
%Then we use
%lemmas \ref{J}, \ref{estimate}(1).
%
%Note that the convergence in the statement of Lemma \ref{J} holds locally uniformly w.r.t.
%$\kappa$.
%
Here
we note that, by Proposition \ref{Jp}, the family
$\{ J^{(N)} (\kappa) \}$
is tight as continuous function-valued process.
Hence
$\tilde{\theta}_n (\kappa) \to \tilde{\theta}_{\infty}(\kappa)$
locally uniformly w.r.t.
$\kappa$.
Since
$\kappa'_k (n) \to \kappa_0$
by
Lemma \ref{estimate}, we have
$\tilde{\theta}_n(\kappa'_k(n))
\to
\tilde{\theta}_{\infty}( \kappa_0 )$.
\QED
%

%%%%%%%%%%%%%%%%%%%%%%%%%%%%%%%%%%%%%%%%%%%%%%%%%%%%%
\section{Deterministic potentials}
\label{intro}
\subsection{Symbolic dynamical systems}
Let
${\cal A} = (a_1, \ldots, a_M)$
be an abstract finite set ("alphabet"),
and consider the
probability spaces
$(\Omega, {\cal F}, {\bf P})$,
$(\Omega_+, {\cal F}_+, {\bf P}_+)$
where
\[
\Omega = {\cal A}^{\bf Z}, \; \Omega_+ = {\cal A}^{\bf N},
\]
${\cal F}$
(resp., ${\cal F}_+$)
is the sigma-algebra generated by the cylinder subsets
\[
C_{i_1, \ldots, i_k}(A_1, \ldots, A_k)
:=
\{\omega: \, \omega_{i_j}\in A_j, \, j=1, \ldots, k\}
\]
(in the case of
${\cal F}_+$,
the indices $i_j$ are non-negative), and
${\bf P}$
(resp., ${\bf P}_+$)
is a probability measure on
${\cal F}$
(resp., on ${\cal F}_+$)
invariant under the shift endomorphism
(isomorphism, in the case of $\Omega$)
$T$
defined by
\[
(T \omega)_i = \omega_{i+1}.
\]
For brevity,
below we often write
$\Omega_\bullet$,
${\cal F}_\bullet$
etc., where
$\bullet$
is "nothing" or "$+$".
In all cases,
we keep the same notation for the shift transformation
$T$.

In a number
of interesting applications, the pair
$({\bf P}_\bullet,T)$
is markovian, i.e.,
${\bf P}_\bullet$
is a Markov measure w.r.t.
$T$:
\[
{\bf P}\{
\omega_{t+1}= a
\,|\,
{\cal F}_{\le t } \}
=
{\bf P}\{
\omega_{t+1} = a \,|\,
{\cal F}_{= t }
\}
\]
where
${\cal F}_{\le t}$
(resp., ${\cal F}_{= t}$)
is generated by the values of the symbols
$\omega_i$
with
$i\le t$
(resp., $i=t$).
Equivalently,
for any
$a, b, b_{-1}, \ldots \in {\cal A}$,
one has
\[
{\bf P}\{
\omega_{t+1} = a
\,|\,
\omega_t = b, \omega_{t-1} = b_{-1}, \ldots
\}
=
{\bf P}\{
\omega_{t+1} = a
\,|\, \omega_t = b
\} =
\Pi_{ab}~,
\]
for some stochastic matrix
$\Pi = (\Pi_{ab})$,
with
$\sum_a \Pi_{ab}=1$.
A particular
subclass of Markov systems is formed by the Bernoulli shifts, where
${\bf P}_\bullet$
is a product measure ${\bf P}_\bullet = \mu^{{\bf Z}_\bullet}$, and $\mu$ is a probability measure
on ${\cal A}$ endowed with the maximal sigma-algebra containing all singletons $\{a\}$, $a\in{\cal A}$.
%
%%%%%%%%%%
\subsection{Symbolic representations for some hyperbolic systems}
%%%%%%%%%
\subsubsection{Dyadic expansion of the unit circle}
\label{sssec:dyadic}
%%%%%
Here
$\Omega = {\bf T}^1 = {\bf R}/{\bf Z}$,
and
we identify it, as a measure (in fact, probability)
space with the interval
$[0,1)$
endowed with the Haar (Lebesgue, in this case) measure
${\bf P}$.
Consider
the measurable transformation
$T: \Omega \to \Omega$
defined by
\[
T: x \mapsto \{2x\} \equiv 2x \pmod 1.
\]
The Lebesgue
measure is
$T$-invariant: for any measurable subset
$A\subset \Omega$,
${\bf P}{ T^{-1} A} = {\bf P}{A}$.
It suffices
to check the latter identity for the intervals
$A = [x,y)$
where it is obvious, since
\[
T^{-1} [x,y) =
\left[ \frac{x}{2}, \frac{y}{2} \right)
\cup
\left[ \frac{x}{2} + \frac 12, \frac{y}{2} + \frac 12 \right),
\]
and each of the two disjoint intervals in the above RHS has length $(y-x)/2$.
Naturally,
$T$
is only an endomorphism, but not isomorphism, for it is not invertible, so it generates a semi-group
$\{T^t, \, t\in {\bf Z}_+\}$.

The standard
symbolic representation for this dynamical system is obtained with the help of the binary expansion of the real numbers
$x\in[0,1)$,
\[
x =
\sum_{i=0}^{\infty}
\frac{\omega_i}{2^{i+1}},
\]
so the identification
$x$
with the infinite word
$(\omega_0, \omega_1, \ldots)\in\{0,1\}^{{\bf Z}_+}$
is a bijection, if one excludes the words having an infinite tail of the form
$(\ldots, \omega_n, 1, 1, 1, \ldots)$,
using the identity
\[
\sum_{i=n+1}^\infty 2^{-i-1} = 2^{-n-1}.
\]
For example,
$(0, 1, 1, 1, \ldots)$
and
$(1, 0, 0, 0, \ldots)$ are two dyadic representations
of the number
$\frac 12$.
We only
define the transformations defined (and, where applicable, invertible) Lebesgue-a.e.

It is
straightforward that
$T$
becomes the left shift on the set of the words
$\omega = (\omega_0, \omega_1, \ldots)$.
%
%%%%%%%%%%%%%%
\subsubsection{Baker's transform}
\label{Baker}
%%%%%%%%%%%%%%
Baker's transform,
or baker's map (N.B.: here "baker" is not a family name  but merely a profession)
is a particular realization of the Bernoulli shift considered in Sect.~\ref{sssec:dyadic}.
From the
symbolic dynamics point of view, it is
obtained from the dyadic expansion of the circle by a canonical procedure extending an endomorphism
(with time given by a semi-group ${\bf N} = {\bf Z}_+$)
to an isomorphism (invertible measure-preserving
transformation with time ${\bf Z}$).
Curiously,
the geometrical realization is quite simple:
$T = {\cal C} \circ {\cal E}$,
where
\[
{\cal E}(x, y) = (2x, y/2)~,
\]
so that
${\cal E}\big([0,1]^2\big) = [0,2]\times [0, \frac 12]$,
and
\begin{equation}\label{eq:def.baker}
{\cal C}(x,y) =
\left\{
\begin{array}{cc}
(x,y), &  (0 \le x < 1)
\\
(x-1, y +\frac 12), & \mbox{ (otherwise)}.
\end{array}
\right.
\end{equation}
The second
stage consists, geometically, in cutting the rectangle
$[0,2]\times [0, \frac 12]$
into two halves, then leaving
$[0,1]\times [0, \frac 12]$
invariant and putting
$[1,2]\times (0, \frac 12]$
on top of the first rectangle.
To obtain
a symbolic dynamics representation
$T_{\cal A}$
of
$T$, use the dyadic expanstions
\[
x = \sum_{i=0}^{+\infty} \omega_i 2^{-i-1}\,, \quad
%%\\
%%%
y = \sum_{i=1}^{+\infty} \omega_{-i} 2^{-i}
\]
and set
\[
\Phi: (x,y) \mapsto (\ldots, \omega_{-2}, \omega_{-1}, \omega_{0}, \omega_{1}, \ldots).
\]
Then
$T_{\cal A} = \Phi \circ T \circ \Phi^{-1}$
is the left shift on infinite words
$\omega\in\{0,1\}^{\bf Z}$.
Indeed,
on the
$x$-coordinate
$T={\cal C}\circ {\cal E}$
acts exactly as the dyadic extension, since
${\cal E}$
acts so, while
${\cal C}$
adds to the
$x$-coordinate either
$0$
or
$-1 = 0 \pmod 1$.
The dyadic
digits of
$y$,
shift to the left, for
${\cal E}$
is multiplication by
$1/2$
in the
$y$-direction;
this determines all digits of the image
$T(x,y)$
with negatives indices except the place no.
$(-1)$.
As to this symbol, the definition
(\ref{eq:def.baker})
clearly shows that
it equals
$0$ if
$x<1/2$,
i.e., if
$\omega_0=0$,
and
$1$
otherwise, so in both cases it is given by
$\omega_0$.

It is readily
seen that
$\omega_0(x,y)$,
as function of the phase point
${\bf u} = (x,y)\in{\bf T}^2$,
is
merely the indicator function of the rectangle
$C_0:=[0,\frac 12)\times [0,1)$.
Respectively,
\[
\omega_t({\bf u}) = 1_{C_0}({\bf u}), \;\; t\in{\bf Z}.
\]
Equivalently,
introducing the partition
$\Omega = C_0 \sqcup C_1$,
$C_1 = \Omega\setminus C_0$,
one can identify the word
$(\omega_t({\bf u}), t\in {\bf Z})$
with the sequence of the ordinal numbers
of the partition elements visited by the trajectory
$\{T^t {\bf u}\}$.
Since
$T$
shrinks the vertical coordinate
$y$
by the factor
$1/2$,
and
$T^{-1}$
does the same to the horizontal coordinate
$x$,
the cylinder sets
\[
\bigcap_{t=-n}^n
\{{\bf u}:\, T^t {\bf u} \in C_{a_t}\}, \;\; a_t \in\{0,1\},
\]
have exponentially decaying diameter as
$n\to \infty$.
%
%%%%%%%%%%%%%
\subsubsection{Algebraic automorphisms of tori}
\label{tori}
%%%%%%%%%%%%
In the general case
$\nu \ge 2$
, the construction of Markov partitions for hyperbolic
toral automorphisms was proposed by Sinai \cite{Sin68}.
This construction
is rather technical
and particularly tedious for the tori of dimension
$\nu >2$
, so we we give only an upshot
in the case
$\nu=2$,
and refer the interested reader to the original paper \cite{Sin68}
and to the books on ergodic theory,
e.g.,
\cite{Bow75,HW09,KFS82,Man87,Shu86}.

Consider
a unimodular matrix with integer entries
\[
M =\left(
       \begin{array}{cc}
         a & b \\
         c & d \\
       \end{array}
     \right) \in SL(2, {\bf Z}),
\]
thus having eigenvalues
$(\lambda_1, \lambda_2) = (\lambda, \lambda^{-1})$,
and assume that
$|\lambda|>1$,
so the
modulus of the second eigenvalue is smaller than $1$.
The most famous example is
\begin{equation}
\label{eq:cat.map}
M =\left(
       \begin{array}{cc}
         2 & 1 \\
         1 & 1 \\
       \end{array}
     \right),
     \quad
     \lambda_{1,2} = \frac{3 \pm \sqrt{5}}{2}.
\end{equation}
Since
$M$,
acting in
${\bf R}^2$
by multiplication, maps the lattice
${\bf Z}^2 \hookrightarrow {\bf R}^2$
into itself, it also acts on the factor-space
${\bf R}^2/{\bf Z}^2 = {\bf T}^2$.
In the case
(\ref{eq:cat.map})
it is usually called Arnold's Cat Map.
The inequalities
$|\lambda^{-1}| < 1 < |\lambda|$
mean that
$M$
is hyperbolic: it has an extending and contracting
eigenspaces. An astute geometrical procedure allows one to partition the torus into a finite
union of parallelepipeds
$C_a$,
$a \in {\cal A} = \{a_1, \ldots, a_N\}$,
with sides parallel to the extending and contracting eigenspaces, in such a way that
\begin{itemize}
  \item
  for a.e.
  ${\bf u} = (x,y)\in {\bf T}^2$, the sequence of symbols $(a_{k(t)}, t\in {\bf Z})$
  such that
\[
T^t {\bf u} \in C_{a_{k(t)}}, \;\; t\in{\bf Z},
\]
  determines the point
  ${\bf u}$
  uniquely; in other words, the torus point
  ${\bf u}$
   is identified
  with the sequence of the ordinal numbers of the parallelepipeds it visits under the dynamics
  $\{T^t\}$;

  \item
  under
  the above identification, the Lebesgue measure on
  ${\bf Z}^2$
  corresponds to a Markov measure w.r.t. the shift
  $T$:
  writing
  ${\bf u} \leftrightarrow (\ldots, \omega_{-1},
  \omega_0, \omega_1, \ldots)$,
  one has
\[
{\bf P}
\{ \omega_{t+1} = a \,|\, \omega_t = b, \omega_{t-1} = b_{-1}, \ldots
\}
=
{\bf P}
\{ \omega_{t+1} = a \,|\, \omega_t = b
\} = \Pi_{ab},
\]
for some irreducible stochastic matrix
$\Pi$.

Such a partition is called a Markov partition for the dynamical system
$(\Omega,{\cal F},{\bf P},T)$.

%%  \item
\end{itemize}
%
%%%%%%%%%%%%
\subsection{Local regularity, quasi-locality and decay of correlations}
%%%%%%%%%%%

For our purposes, the key feature of the Markov partitions is exponential decay
of the diameter of a cylinder set in the torus
\[
{\cal X}_{(-n, \dots, n)}(\alpha_{-n}, \ldots \alpha_n)
=
\{x(\omega) \in {\bf T}^\nu:\, \omega_i = \alpha_{j}, -n \le i \le n \}
\]
as
$n\to +\infty$.
The geometrical
mechanism of this decay is essentially the same as for the
baker's map, although the decay exponent is determined by the eigenvalues of the generating
linear map
${\cal L}$.
Therefore,
for any two points
$x,y$
whose symbolic representations ("letters"
$\omega_i(x)$
and
$\omega_i(y)$)
agree on a long interval of indices,
$i\in\{-n -n+1, \ldots, n\}$,
we have
\begin{equation}
\label{eq:dist.x.y}
\mbox{dist}_{{\bf T}^\nu}(x,y) \le q^n, \; q\in(0,1)~.
\end{equation}
Consequently,
for any continuous function
$f:\,{\bf T}^\nu\to {\bf R}$
with continuity modulus
\[
s_f(\epsilon) :=
\sup_{
y\in B_{\epsilon}(x)
} \,
\big| f(y) - f(x)\big|
\]
one has for the points
$x,y$
satisfying
(\ref{eq:dist.x.y})
\[
\big| f(y) - f(x)\big|
\le
s_f( q^n)~.
\]
In particular, for any H\"{o}lder continuous function of order $\beta\in(0,1]$, we have
$$
\big| f(y) - f(x)\big| \le C \tilde q^n~, \;\; \tilde q = q^\beta \in(0,1).
$$

In all
considered examples, the existence of a Markov partition gives rise to the quasi-locality
of the deterministic random potentials as functions of symbols in the infinite words
$(\omega_i)_{i\in {\bf Z}_\bullet}$.

Introduce
the following notation: for a function
$f: {\cal A}^{{\bf Z}_\bullet} \to {\bf R}$
and an index subset
$I\subset {\bf Z}_\bullet$,
let
\begin{equation}
\mbox{Var}_I (f) :=
\sup_{ \omega', \omega:\,
\pi_{I}\omega'
=
\pi_{I}\omega}
\big| f(\omega') - f(\omega) \big| ,
\end{equation}
where
$\pi_I(\omega)$
is the finite sub-word of
$\omega$
formed by the letters
$(\omega_i, \, i\in I)$.

\begin{definition}
A function
$f: {\cal A}^{{\bf Z}_\bullet} \to {\bf R}$
is called quasi-local if there exists
$C=C(f)>0$
and
$q\in(0,1)$
such that for any $
n\ge 1$
and any finite word
$(\omega_{-n}, \omega_{-n+1}, \ldots, \omega_n)$
\begin{equation}
\mbox{Var}_{[-m,n]}(f)
\le C q^{m\wedge n}.
\end{equation}
\end{definition}

In turn,
the quasi-locality implies exponential decay
of correlations (this decay may be slower for the sampling functions
$f$
featuring lower regularity that
H\"{o}lder continuity).

The bottom line
is that in the above mentioned examples of hyperbolic dynamical systems
on tori
${\bf T}^\nu \cong [0,1)^\nu \subset {\bf R}^\nu$,
the corresponding deterministic potentials
feature a fast decay of correlations sufficient for the extension of the Kolmogorov's
connection between the convergence in mean square and the a.s. convergence of the random series
\[
\sum_{j=1}^\infty \frac{v_j}{j^\alpha},
\;\;
\alpha\in \left( \frac 12, 1 \right].
\]
%

%%%%%% Acknowledgement %%%%%%%%%%%%%%%%%%%%%%%%
\noindent {\bf Acknowledgement }
The authors
would like to thank the Isaac Newton Institute for
Mathematical Sciences for its hospitality during the programme
``Periodic and Ergodic Spectral Problems"
supported by EPSRC Grant Number EP/K032208/1.
This work is partially supported by
JSPS KAKENHI Grant
Number 26400145(F.N.)
%

%%%%% REFERENCES %%%%%%%%%%%%%%%%%%%%%
%

%

%\end{thebibliography}


\small
\begin{thebibliography}{99}
%

\bibitem{ALS}
Avila, A.,  Last, Y., Simon, B.:
Bulk Universality and Clock Spacing of zeros for Ergodic Jacobi Matrices with A.C. spectrum,
Anal. PDE {\bf 3}(2010), 81-108.
%

\bibitem
{Bow75}
Bowen, R.: 
Equilibrium states and the ergodic theory of Anosov Diffeomorphisms.
Lecture Notes in Mathematics, \textbf{470}, Springer, Berlin (1975)
%

%
\bibitem{CS}
Chulaevsky, V., Spencer, T.:
Positive Lyapunov exponents for a class of deterministic potentials,
Comm. Math. Phys. {\bf 168}, No.3 (1995), 455-466.
%

\bibitem
{HW09}
Hu,T.-C., Weber, N.C.: 
A note on strong convergence of sums of dependent random variables.
J. Prob. Stat., \textbf{2009}, 873274 (2009)

%
\bibitem{KS}
Killip, R., Stoiciu, M.:
Eigenvalue statistics for CMV matrices :
from Poisson to clock via random matrix ensembles,
Duke Math. {\bf 146}(2009), 361-399.
%


\bibitem{KLS}
Kiselev, A., Last, Y., Simon, B.: 
Modified Pr\"ufer and EFGP Transforms 
and the Spectral Analysis of One-Dimensional Schr\"odinger Operators, 
Commun. Math. Phys. {\bf 194}(1997), 1-45. 
%

\bibitem
%[Sim82]
{KFS82}
Kornfeld, I.P., Fomin, S. V., Sinai, Ya. G.: 
Ergodic theory.
Gruhndlehren der mathematischen Wissenschaften. Springer (1982)


\bibitem{KN1}
Kotani, S., Nakano, F.:
Level statistics for the one dimensional Schroedinger operators with random decaying potential,
Interdisciplinary Mathematical Sciences Vol. 17 (2014), p.343-373.


\bibitem{KN2}
Kotani, S., Nakano, F.:
Poisson statistics for 1d  Schr\"odinger operators with random decaying potential,
arXiv:1605.02416.

%
\bibitem{KU}
Kotani, S. Ushiroya, N.:
One-dimensional Schr\"odinger operators with random decaying
potentials,
Comm. Math. Phys. {\bf 115}(1988), 247-266.
%

%
\bibitem{KVV}
Kritchevski, E., Valk\'o, B., Vir\'ag, B.: 
The scaling limit of the critical one-dimensional random 
Sdhr\"odinger operators, Commun. Math. Phys. {\bf 314}(2012), 775-806.
%


%
\bibitem{MD}
Mallik, A., Dolai, D.:
Spectral Statistics for one dimensional Anderson model with unbounded but decaying potential,
arXiv : 1602.02986.
%

\bibitem
{Man87}
Ma\~n\'{e}, R.: 
Ergodic theory and differentiable dynamics.
Springer, Berlin (1987)


\bibitem{N}
Nakano, F.:
Level statistics for one-dimensional Schr\"odinger operators and Gaussian beta ensemble,
J. Stat. Phys. {\bf 156}(2014), 66-93.

\bibitem
{Shu86}
Shub, M.: 
Global stability of dynamical systems.
Springer, Berlin (1986)

\bibitem
%[Sim82]
{Sin68}
Sinai, Ya. G.: 
Construction of Markov partitions.
\emph{Funkt. Anal. i Prilozhen.}
\textbf{2}, no. 3,  70--80 (1968)
%




\end{thebibliography}
\end{document}